\documentclass[11pt]{article}
\usepackage{amsthm,amsfonts,amsmath,amssymb,amscd, accents, mathrsfs, mathtools}
\usepackage{authblk}
\usepackage{hyperref}
\usepackage{multicol}
                                \usepackage{verbatim}
\swapnumbers
                              
\mathtoolsset{showonlyrefs=false}
                                
\theoremstyle{plain}

\usepackage{enumitem}
\numberwithin{equation}{section}
\newtheorem{theorem}{Theorem}[section]

\newtheorem{assumption}[theorem]{Assumption}

\newtheorem{definition}[theorem]{Definition}
\newtheorem{notation}[theorem]{Notation}
\newtheorem{example}[theorem]{Example}

\theoremstyle{remark}
\newtheorem{remark}[theorem]{Remark}
\numberwithin{equation}{section}

\newcommand{\bR}{{\mathbb R}}

\newcommand{\cR}{{\mathcal R}}

\newcommand{\cH}{{\mathcal H}}

\newcommand{\cC}{{\mathcal C}}

\newcommand{\cI}{{\mathcal I}}

\newcommand{\cB}{{\mathcal B}}
\newcommand{\cE}{{\mathcal E}}

\newcommand{\ket}[1]{\left\vert #1\right\rangle}
\newcommand{\bra}[1]{\left\langle #1\right\vert}

\def\idty{{\mathchoice {\mathrm{1\mskip-4mu l}} {\mathrm{1\mskip-4mu l}} %
{\mathrm{1\mskip-4.5mu l}} {\mathrm{1\mskip-5mu l}}}}

\newcommand{\Tr}{\mathrm{Tr}}

\newcommand{\supp}{\mathop{\mathrm{supp}}}

\newcommand{\be}{\begin{equation}}
\newcommand{\ee}{\end{equation}}
\newcommand{\bea}{\begin{eqnarray}}
\newcommand{\eea}{\end{eqnarray}}
\newcommand{\beann}{\begin{eqnarray*}}
\newcommand{\eeann}{\end{eqnarray*}}

\usepackage{color}

\begin{document}

\title{Measure of genuine coherence based of quasi-relative entropy}
\author{Anna Vershynina}
\affil{\small{Department of Mathematics, Philip Guthrie Hoffman Hall, University of Houston, 
3551 Cullen Blvd., Houston, TX 77204-3008, USA}}
\renewcommand\Authands{ and }
\renewcommand\Affilfont{\itshape\small}

\date{\today}

\maketitle

\begin{abstract} 
We present a genuine coherence measure based on a quasi-relative entropy as a difference between quasi-entropies of the dephased and the original states. The measure satisfies non-negativity and monotonicity under genuine incoherent operations (GIO). It is strongly monotone under GIO in two- and three-dimensions, or for pure states in any dimension, making it a genuine coherence monotone. We provide a bound on the error term in the monotonicity relation under GIO in terms of the trace distance between the original and the dephased states. Moreover, the lower bound on the coherence measure can also be calculated in terms of this trace distance.
\end{abstract}
\begin{multicols}{2}

\section{Introduction}
Quantum coherence is a fundamental property of quantum systems, describing the existence of quantum interference. It is widely used in thermodynamics \cite{A14, C15, L15}, transport theory \cite{RM09, WM13}, and quantum optics \cite{G63, SZ97}, among few applications. Recently, problems involving coherence included quantification of coherence \cite{BC14, PC16, RPL16, R16, SX15, YZ16}, distribution \cite{RPJ16}, entanglement \cite{CH16, SS15}, operational resource theory \cite{CG16, CH16, DBG15, WY16},  correlations \cite{HH18, MY16, TK16}, with only a few references mentioned in each. See \cite{SAP17} for a more detailed review.

As a golden standard it is taken that any ``good" coherence measure should satisfy four criteria presented in \cite{BC14}: vanishing on incoherent states; monotonicity under  incoherent operations; strong monotonicity under incoherent operations, and convexity. Alternatively, the last two properties can be substituted by an additivity for subspace independent states, which was shown in \cite{YZ16}. 

 A number of ways has been proposed as a coherence measure, but only a few satisfy all necessary criteria \cite{BC14, ZY18, Zetal17}. A broad class of coherence measure are defined as the minimal distance $D$ to the set of incoherent states $\cI$, as
 $$\cC_D(\rho)=\min_{\delta\in \cI}D(\rho, \delta). $$
In \cite{BC14}, it was shown that coherence vanishes on incoherent states when the distance vanishes only on identical states; the measure is monotone when the distance is contractive under quantum channels; and it is convex when the distance is jointly convex. Strong monotonicity property is more challenging to pinpoint. Measures that satisfy the strong monotonicity that have been introduced up to date, are based on $l_1$, relative entropy, Tsallis entropy, and real symmetric concave functions on a probability simplex. 

Another approach to generate physically relevant coherence measures is to consider different incoherent operations. The largest class of incoherent operations is called maximally incoherent (MIO), and it consists of all completely-positive trace-preserving (CPTP) maps that preserve the set of incoherent states. The smaller set, called incoherent operations (IO)  \cite{BC14}, has Krauss operators that each preserve the set of incoherent states (see Definition \ref{def:IO}). A smaller set consists of strictly incoherent operations (SIO) \cite{WY16, YMGGV16}, which are the result of action on a primary and ancillary systems that do not generate coherence on a primary system, see Definition \ref{def:SIO}. And the last class of operations that is discussed in this paper, is called genuine incoherent operations (GIO) \cite{DS16}, which act trivially on incoherent states, see Definition \ref{def:GIO}. See \cite{CG16-2} for a larger list of incoherent operations, and their comparison. For these types of incoherent operations one may look at similar properties as the ones presented in \cite{BC14}. Restricted to GIO, one would obtain a measure of genuine coherence when it is non-negative and monotone, or a coherence monotone when it is also strongly monotone under GIO.

In \cite{DS16}, the following genuine coherence measure was proposed:
$$\cC_D(\rho)=D(\rho\|\Delta(\rho))\ , $$
for a distance $D$, and $\Delta(\rho)$ being the dephased state in a pre-fixed basis, see   Notation \ref{notation}. It was shown that this is a genuine coherence measure if the distance is contractive under unital operations. If fact, the monotonicity holds not only for GIO maps but for dephasing-covariant incoherent operations (DIO) as well (the ones that commute with the dephasing operator).

Here we propose another genuine coherence measure based on a  quasi-relative entropy:
$$\cC_f(\rho)=S_f(\Delta(\rho))-S_f(\rho)\ , $$
here $S_f(\rho)$ is a quasi entropy, which could be defined in two ways, one of which is $S_f(\rho)=-S_f(\rho\|| I)$. The motivation for this definition comes from the relative entropy coherence. It was shown \cite{BC14} that for a relative entropy $S(\cdot\|\cdot)$, there is a closed expression of a distance-based coherence measure:
$$\min_{\delta\in \cI}S(\rho\|\delta)=S(\rho\|\Delta(\rho))=S(\Delta(\rho))-S(\rho)\ . $$

In general, for quasi-relative entropies neither of these equalities will hold. This can be seen for Tsallis relative entropy, which is a particular case of a quasi-relative entropy. The closest incoherent state is given in \cite{R16}, and it is not a dephased state $\Delta(\rho)$. The second equality does not hold either in general.

We show that quasi-relative entropy coherence, which we call $f$-coherence, is unique for pure states, non-negative, zero if and only if a state is incoherent, and monotone under GIO maps. Moreover, we give a lower bound on this coherence in terms of a trace distance between a state and its dephased state, we provide an if and only if condition on a GIO map that saturates the monotonicity relation, and bound the error term in the monotonicity relation. Additionally, we investigate when the $f$-coherence would be monotone under a larger class of SIO maps. 

We show that $f$-coherence saturates strong monotonicity under GIO maps in two- and three-dimensions, and it satisfies the strong monotonicity under GIO maps in any dimensions for pure states.

\section{Preliminaries}

\subsection{Coherence}

Let $\cH$ be a $d$-dimensional Hilbert space. Let us fix a basis $\cE=\{\ket{j}\}_{j=1}^d$ of vectors in $\cH$.
\begin{definition} A state $\delta$ is called {\it incoherent} if it can be represented as follows
$$\delta=\sum_j \delta_j\ket{j}\bra{j}\ . $$
\end{definition}

\begin{notation}\label{notation}
Denote the set of {\bf incoherent states} for a fixed basis $\cE=\{\ket{j}\}_j$ as $$\cI=\{\rho=\sum_jp_j\ket{j}\bra{j}\}\ .$$
A {\bf dephasing} operation in $\cE$ basis is the following map:
$$\Delta(\rho)=\sum_j \bra{j}\rho\bra{j} \ket{j}\bra{j}\ . $$
\end{notation}

\begin{definition}\label{def:IO} A CPTP map $\Phi$ with the following Kraus operators
$$\Phi(\rho)=\sum_n K_n \rho K_n^*\ , $$
is called {\bf the incoherent operation (IO)} or incoherent CPTP (ICPTP), when the Kraus operators satisfy
$$K_n \cI K_n^*\subset \cI,\ \text{for all }n \ , $$
besides the regular completeness relation $\sum_n K_n^*K_n=\idty$.
\end{definition}

Any reasonable measure of coherence $\cC(\rho)$ should satisfy the following conditions
\begin{itemize}
\item (C1) $\cC(\rho)=0$ if and only if $\rho\in\cI$;
\item (C2) Non-selective monotonicity under IO maps (monotonicity)
$$\cC(\rho)\geq \cC(\Phi(\rho))\ ; $$
\item (C3) Selective monotonicity under IO maps (strong monotonicity)
$$\cC(\rho)\geq \sum_n p_n \cC(\rho_n)\ , $$
where $p_n$ and $\rho_n$ are the outcomes and post-measurement states
$$\rho_n=\frac{K_n\rho K_n^*}{p_n},\ \ p_n=\Tr K_n\rho K_n^*\ . $$
\item (C4) Convexity, 
$$\sum_n p_n \cC(\rho_n)\geq \cC\left(\sum_n p_n\rho_n\right)\ , $$
for any sets of states $\{\rho_n\}$ and any probability distribution $\{p_n\}$.
\end{itemize}


These properties are parallel with the entanglement measure theory, where the average entanglement is not increased under the local operations and classical communication (LOCC). Notice that coherence measures that satisfy conditions (C3) and (C4) also satisfies condition (C2). 

In \cite{DS16} a class of incoherence operations was defined, called {genuinely incoherent operations (GIO)} as quantum operations that preserve all incoherent states.

\begin{definition}\label{def:GIO} An IO map $\Lambda$ is called a {\bf genuinely incoherent operation (GIO)} is for any incoherent state $\delta\in\cI$,
$$\Lambda(\delta)=\delta\ . $$
\end{definition}

An operation $\Lambda$ is GIO if and only if all Kraus representations of $\Lambda$ has all Kraus operators diagonal in a pre-fixed basis  \cite{DS16}. 

Conditions (C2), (C3) and (C4) can be restricted to GIO maps to obtain different classes of coherence measures.
\begin{definition}
In this case, a {\bf measure of genuine coherence} satisfies at least (G1) and (G2). And if a coherence measure fulfills conditions (G1), (G2), (G3) it is called {\bf genuine coherence monotone}. 
\end{definition}

A larger class of IO maps was defined in \cite{WY16, YMGGV16}.
\begin{definition}\label{def:SIO}
An IO map $\Lambda$ is called {\bf strictly incoherent operations (SIO)} if its Kraus representation operator commute with dephasing, i.e. for $\Lambda(\rho)=\sum_j K_j\rho K_j^*$, we have for any $j$,
$$K_j\Delta(\rho)K_j^*=\Delta(K_j\rho K_j^*) \ . $$
\end{definition}

Since Kraus operators of GIO maps are diagonal in $\cE$ basis, any GIO map is SIO as well, i.e. GIO $\subset$ SIO, \cite{DS16}. 

One may consider an additional property, closely related to the entanglement theory:
\begin{itemize}
\item (C5) Uniqueness for pure states: for any pure state $\ket{\psi}$ coherence takes the form:
$$\cC(\psi)=S(\Delta(\psi))\ , $$
where $S$ is the von Neumann entropy and $\Delta$ is the dephasing operation defined as
$$\Delta(\rho)=\sum_j \bra{j}\rho\ket{j}\ket{j}\bra{j}\ . $$
\end{itemize}





\subsection{Quasi-relative entropy}
Quantum quasi-relative entropy was introduced by Petz \cite{P85, P86} as a quantum generalization of a classical Csisz\'ar's $f$-divergence \cite{C67}. It is defined in the context of von Neumann algebras, but we  consider only the Hilbert space setup. Let $\cH$ be a  finite-dimensional  Hilbert space, and $\rho$ and $\sigma$ be two states (given by density operators).

\begin{definition}\label{def:qre}
For strictly positive bounded operators $A$ and $B$ acting on a finite-dimensional Hilbert space $\cH$, and for any {continuous} function $f: (0,\infty)\rightarrow \bR$, {\it the quasi-relative entropy} (or sometimes referred to as {\it the $f$-divergence}) is defined as 
$$S_f(A|| B)=\Tr(f(L_BR_A^{-1}){A})\ ,$$
where left and right multiplication operators are defined as $L_B(X)=BX$ and $R_A(X)=XA$. 
\end{definition}

There is a straightforward way to calculate the quasi-relative entropy from the spectral decomposition of operators \cite{HM17, V16}. Let $A$ and $B$ have the following spectral decomposition
\begin{equation}\label{eq:spectral}
A=\sum_j\lambda_j\ket{\phi_j}\bra{\phi_j}, \ \ B=\sum_k\mu_k\ket{\psi_k}\bra{\psi_k}\ .
\end{equation}
 the set $\{\ket{\phi_k}\bra{\psi_j}\}_{j,k}$ forms an orthonormal basis of $\cB(\cH)$, the space of bounded linear operators, with respect to the Hilbert-Schmidt inner product defined as $\langle A, B \rangle=\Tr(A^*B)$. By \cite{V16}, the product of left and right multiplication operators can be written as
\begin{equation}\label{eq:modular}
L_BR_A^{-1}=\sum_{j,k} \frac{\mu_k}{\lambda_j}P_{j,k}\ ,
\end{equation}
where $P_{j,k}:\cB(\cH)\rightarrow\cB(\cH)$ is defined by
$$P_{j,k}(X)=\ket{\psi_k}\bra{\phi_j}\bra{\psi_k}X\ket{\phi_j}\ . $$
The quasi-relative entropy is calculated as follows
\begin{equation}\label{eq:formula}
S_f(A||B)=\sum_{j,k}\lambda_j f\left(\frac{\mu_k}{\lambda_j}\right)|\bra{\psi_k}\ket{\phi_j}|^2\ . 
\end{equation}

\begin{theorem}\label{thm:pos-f-div}
(\cite{P85}) For states, i.e. trace one positive density matrices $\rho$ and $\sigma$, the quasi-relative entropy is bounded below by  
$$S_f(\rho\|\sigma)\geq f(1).$$
The equality happens for a non-linear function $f$ if and only if $\rho=\sigma$.
\end{theorem}

It is natural to require the quasi-relative entropy to be zero for equal state, and therefore we assume throughout the paper that $f(1)=0$.

For an operator convex function, $f$, the quasi-relative entropy is jointly convex and monotone under CPTP maps \cite{HM17}. The equality in monotonicity holds if and only if the map is reversible on these two states, i.e. for two states $\rho$ and $\sigma$ with $\supp \rho\subset \supp\sigma$, and a CPTP map $\Lambda$, the equality
$$S_f(\rho\|\sigma)=S_f(\Lambda(\rho)\|\Lambda(\sigma)) $$
is satisfied if and only if 
$$\cR_\sigma(\Lambda(\rho))=\rho\ , $$
where $\cR_\sigma$ is the Petz's recovery map defined as
\begin{equation}\label{eq:Petz-map}
\cR_\sigma(\omega)=\sigma^{1/2}\Lambda^*\left(\Lambda(\sigma)^{-1/2}\omega\Lambda(\sigma)^{-1/2} \right)\sigma^{1/2}\ .
\end{equation}

\begin{assumption}
Throughout the paper we will assume that the function $f$ is operator convex and $f(1)=0$.
\end{assumption}

For any function $f$, its transpose $\tilde{f}$ is defined as
$$\tilde{f}(x)=xf\left(\frac{1}{x} \right)\ , \ x\in(0.\infty)\ . $$
The transpose $\tilde{f}$ of an operator convex function $f$ on $(0,\infty)$ is operator convex again, \cite{HM17}. From (\ref{eq:formula}) it follows that
$$S_{\tilde{f}}(\rho\|\sigma)=S_f(\sigma\|\rho)\ . $$

\begin{example} For $f(x)=-\log x$, the quasi-relative entropy becomes the Umegaki relative entropy
$$S_{-\log}(\rho\|\sigma)=S(\rho\|\sigma)=\Tr (\rho\log\rho-\rho\log\sigma)\ . $$
\end{example}

\begin{example}
For {$p\in(-1,1)$} and $p\neq 0$ let us take the function 
$$f_p(x):=\frac{1}{p(1-p)}(1-x^p)\ ,$$
which is {operator} convex. The quasi-relative entropy for this function is calculated to be
$$S_{f_p}(\rho|| \sigma)=\frac{1}{p(1-p)}\left(1-\Tr(\sigma^{p}\rho^{1-p})\right)\ .$$
\end{example}

\begin{example}
For $p\in(-1,1)$ take $q=1-p\in(0,2)$, the function
$$f_q(x)=\frac{1}{1-q}(1-x^{1-q}) $$
is operator convex. The quasi-relative entropy for this function is known as Tsallis $q$-entropy
$$S_q(\rho\|\sigma)= \frac{1}{1-q}\left(1-\Tr(\rho^{q}\sigma^{1-q})\right)\ .$$
\end{example}


\section{$f$-entropy}
For a convex, operator monotone decreasing function $f$, such that $f(1)=0$, define entropy two ways.
\begin{definition} The $f$-entropy is defied as
\begin{equation}\label{eq:f-entropy}
S_f(\rho):=f(1/d)-S_f(\rho\|I/d)\ .
\end{equation}
\begin{equation}\label{eq:f-entropy-hat}
\hat{S}_f(\rho):=-S_f(\rho\|I)\ .
\end{equation}
Let us use a notation $\tilde{S}_f$ for either $S_f$ or $\hat{S}_f$.
\end{definition}

\begin{theorem}
$f$-entropy is non-negative, and is zero on pure states.
\end{theorem}
\begin{proof}
Let $\{\lambda_j\}$ be the eigenvalues of $\rho$. Then from (\ref{eq:formula}) we have
\begin{equation}\label{eq:entropy-eigen}
S_f(\rho)=f(1/d)-\sum_j \lambda_j f\left(\frac{1}{d\lambda_j}\right)\ ,
\end{equation}
and \begin{equation}\label{eq:entropy-eigen-hat}
\hat{S}_f(\rho)=-\sum_j \lambda_j f\left(\frac{1}{\lambda_j}\right)\ .
\end{equation}
A sequence of eigenvalues $\{\lambda_j\}$ is majorized by a sequence $\{1,0,\dots, 0\}$. Since a perspective function (or a transpose function) $xf(1/x)$ is convex for a convex function $f$ \cite{HM17}, this implies that by results on Schur-concavity \cite{HLP29, MOA10, S23} {we have}
$$\sum_j \lambda_j f\left(\frac{1}{d\lambda_j}\right)\leq f(1/d)\ . $$
Here, if needed, we adopt a convention $0\cdot \pm\infty:=0$ \cite{HM11}.

Since $f$ is monotonically decreasing and $f(1)=0$, for any $0\leq\lambda_j\leq 1$, $ f\left(\frac{1}{\lambda_j}\right)\leq 0$. Thus, $\hat{S}_f\geq 0.$

When $\rho=\ket{\Psi}\bra{\Psi}$ is a pure state, there is only one eigenvalue $\lambda=1$. Then
$$S_f(\ket{\Psi}\bra{\Psi})=f(1/d)-f(1/d)=0\ , $$
and 
$$\hat{S}_f(\ket{\Psi}\bra{\Psi})=-f(1)=0\ . $$
\end{proof}

\begin{theorem}\label{thm:entropy-max}
The maximum value of $f$-entropy is reached on the maximally mixed state $I/d$ and it is
$$S_f(\rho)\leq f(1/d) \ , $$
and
$$\hat{S}_f(\rho)\leq -f(d)\ . $$
\end{theorem}
\begin{proof} From Theorem \ref{thm:pos-f-div}, $S_f(\rho\|I/d)\geq 0$, or since $f$ is convex, we have
$$\sum_j \lambda_j f\left(\frac{1}{d\lambda_j}\right)\geq f\left(\sum_j \lambda_j \frac{1}{d\lambda_j} \right)=f(1)=0\ . $$
 Similarly,
$$\sum_j \lambda_j f\left(\frac{1}{\lambda_j}\right)\geq f\left(\sum_j \lambda_j \frac{1}{\lambda_j} \right)=f(d)\ . $$
From (\ref{eq:entropy-eigen}) and  (\ref{eq:entropy-eigen-hat}), the result follows. 
Clearly, when $\rho=I/d$, we have $S_f(I/d)=f(1/d)-0=f(1/d)$, and from (\ref{eq:entropy-eigen-hat}) we have $\hat{S}_f(I/d)=-f(d).$
\end{proof}

\begin{theorem}
The $f$-entropies are concave in $\rho$. Let $\{p_k\}$ be a probability distribution and $\rho_k$ be some states, then for $\rho=\sum_k p_k\rho_k$, we have
$$\tilde{S}_f(\rho)\geq \sum_kp_k\tilde{S}_f(\rho_k)\ . $$
\end{theorem}
\begin{proof}
This immediately follows from the joint convexity of $f$-divergence \cite{HM11, HM17}.
\end{proof}

\begin{theorem}
The $f$-entropies are invariant under unitaries.
\end{theorem}
\begin{proof}
Since a unitary operation $U\rho U^*$ does not change the eigenvalues of $\rho$, and the $f$-entropies are the functions of eigenvalues of $\rho$, this implies that $f$-entropies are invariant under any operations that preserve eigenvalues.
\end{proof}

\begin{theorem}\label{thm:entropy-mono}
The $f$-entropies are non-decreasing under untial CPTP maps, i.e. for any linear CPTP map $\Lambda$, such that $\Lambda(I)=I$, we have
$$\tilde{S}_f(\Lambda(\rho))\geq \tilde{S}_f(\rho)\ . $$
\end{theorem}
\begin{proof}
Let us denote $\sigma= I$ or $\sigma=I/d$, which corresponds to the appropriate $f$-entropy. Then
\begin{align}
&\tilde{S}_f(\Lambda(\rho))-\tilde{S}_f(\rho)=S_f(\rho\|\sigma)-S_f(\Lambda(\rho)\|\sigma)\\
&=S_f(\rho\|\sigma)-S_f(\Lambda(\rho)\|\Lambda(\sigma))\geq 0 \label{eq:entropy-mono-eq}\ .
\end{align}
The last equality holds since $\Lambda$ is unital, and the inequality holds due to the monotonicity of $f$-divergence under CPTP maps  \cite{LR99, P85, TCR09}.
\end{proof}

\section{Measure of genuine coherence}

In a $d$-dimensional Hilbert space $\cH$, fix a basis $\cE=\{\ket{j}\}_{j=0}^{d-1}$. 

\begin{definition}
For any entropy function $S$, which is non-decreasing under CPTP maps, define coherence as follows:
\begin{equation}\label{eq:def-coh}
\cC_S(\rho):=S(\Delta(\rho))-S(\rho)\ .
\end{equation}
\end{definition}
In particular, for any operator convex and operator monotone decreasing function $f$, define two $f$-coherence measures.
\begin{definition}
For entropy defined in (\ref{eq:f-entropy}),
\begin{equation}\label{eq:def-coh-f}
\cC_f(\rho):=S_f(\Delta(\rho))-S_f(\rho)\ .
\end{equation}
For entropy defined in (\ref{eq:f-entropy-hat}),
\begin{equation}\label{eq:def-coh-f-hat}
\hat{\cC}_f(\rho):=\hat{S}_f(\Delta(\rho))-\hat{S}_f(\rho)\ .
\end{equation}
Let us denote $\tilde{\cC}_f$ as either one ${\cC}_f$ or $\hat{\cC}_f$ for shortness.
\end{definition}

If $\{\lambda_j\}$ are the eigenvalues of $\rho$, and the diagonal elements of $\rho$ in $\cE$ basis are $\chi_j=\bra{j}\rho\ket{j}$, then from (\ref{eq:entropy-eigen}) and (\ref{eq:entropy-eigen-hat}), we have
\begin{equation}\label{eq:c_f-l-c}
\cC_f(\rho)=\sum_j \lambda_j f\left(\frac{1}{d\lambda_j}\right)-\sum_j \chi_j f\left(\frac{1}{d\chi_j}\right)\ ,
\end{equation}
and
\begin{equation}\label{eq:c_f-l-c-hat}
\hat{\cC}_f(\rho)=\sum_j \lambda_j f\left(\frac{1}{\lambda_j}\right)-\sum_j \chi_j f\left(\frac{1}{\chi_j}\right)\ .
\end{equation}

\subsection{Example}
\subsubsection{Log}
Since $f(x)=-\log(x)$ is operator convex, coherence measure defined above coincides with \cite{BC14}:
\begin{align}
\cC(\rho)=\hat{\cC}_f(\rho)&=S_{\log}(\Delta(\rho))-S_{\log}(\rho)\\
&=\sum_j \lambda_j\log\lambda_j-\sum_j\chi_j\log\chi_j\\
&=S(\Delta(\rho))-S(\rho)\\
&=S(\rho\|\Delta(\rho))\\
&=\min_{\delta\in \cI}S(\rho\|\delta)\ .
\end{align}

\subsubsection{Power}
The function $f(x)=\frac{1}{1-\alpha}(1-x^{1-\alpha})$ is operator convex for $\alpha\in(0,2)$. The coherence monotone is then defined as
\begin{equation}
\cC_\alpha(\rho)=\frac{d^{\alpha-1}}{1-\alpha} \left[\sum_j \chi_j^\alpha-\sum_j\lambda_j^\alpha \right]=d^{\alpha-1}\hat{\cC}_\alpha(\rho)\ .
\end{equation}

\section{Properties}
\subsection{Uniqueness for pure states.} For any pure state coherence becomes an entropy of a dephased state:
$$\cC_S(\psi)=S(\Delta(\psi))\ . $$
This holds since entropies are zero on pure states.

\subsection{Positivity}
\begin{theorem}
$\cC_S$ and, in particular, $\tilde{\cC}_f$ are non-negative.
\end{theorem}
\begin{proof}
By assumption $S$ is non-decreasing under CPTP maps, it follows that $\cC$ is non-negative. 
 
This holds for $f$-entropies as well due to Theorem \ref{thm:entropy-mono}, since the dephasing operation is unital.
\end{proof}

Clearly, for any incoherent state $\rho$, coherence $\cC_S(\rho)=0$. Having no information on the saturation condition for a general entropy $S$, it is impossible to say what happens in the other direction.  Consider $f$-coherences (\ref{eq:def-coh-f}) and  (\ref{eq:def-coh-f-hat}). 

\begin{theorem}
$\tilde{\cC}_f(\rho)=0$ if and only if $\rho\in\cI$ is incoherent state.
\end{theorem}
\begin{proof}
In Theorem \ref{thm:entropy-mono}, the equality in the only inequality (\ref{eq:entropy-mono-eq}) holds if and only if there is a recovery map $\cR$  such that $\cR(\Delta(\rho))=\rho$ and $\cR(I)=I$, \cite{HM11, HM17}. By (\ref{eq:Petz-map}), this map admits the following explicit form: denoting $\sigma=I$
$$\cR_\sigma(\omega)=\sigma^{1/2}\Delta^*\left( \Delta(\sigma)^{-1/2}\omega\Delta(\sigma)^{-1/2}\right)\sigma^{1/2}\ , $$
where $\Delta^*$ is a dual map of $\Delta.$
Since $\Delta$ is a linear unital GIO map, we have
\begin{equation}
\cR_\sigma(\omega)= \Delta^*(\omega)\ .
\end{equation}
Therefore, condition $\cR_\sigma(\Delta(\rho))=\rho$ implies that
\begin{align}
\rho=\Delta^*(\Delta(\rho))\ .
\end{align}
Since $\Delta^*=\Delta$, we have that  $\rho=\Delta(\rho)$, which happens if and only if $\rho\in \cI$. Thus, $\cC_f(\rho)=0=\hat{\cC}_f(\rho)$ if and only if $\rho\in\cI$.
\end{proof}

A strengthening of the monotonicity inequality for $f$-divergence was presented in \cite{CV18}. Using this result, we {obtain} the following lower bound on $f$-coherence.
\begin{theorem} Let $f$ be an operator monotone 
decreasing function, and $T>0$.  Suppose for some constant $c> 0$, there is a constant $C> 0$ so that
${\rm d}t\leq  CT^{2c}{\rm d}\mu_f(t)$ for 
$t\in[T^{-1}, T]$.  Then there is an explicitly computable  constant 
$K_f(\rho)$ depending only on the smallest non-zero eigenvalue of $\rho$, $C$ and $c$, such that, 
\begin{equation}
\cC_f(\rho)\geq K_f(\rho) \|\rho-\Delta(\rho)\|_1^{4(1+c)}\ .
\end{equation}
Here, $\|A\|_1=\Tr|A|=\Tr\sqrt{A^*A}$ is the trace-norm of an operator.
\end{theorem}
From this inequality, the above condition of a zero coherence becomes apparent, i.e. $\cC_f(\rho)=0$ if and only if $\rho\in\cI$. 

The upper bound given below extends the upper bound for a relative entropy of coherence \cite{BC14} to any $f$-coherence.
\begin{theorem}
The coherence is upper bounded by 
$$\cC_f(\rho)\leq f(1/d)\ , $$
and
$$\hat{\cC}_f(\rho)\leq -f(d) \ . $$
The maximum value is reached for a maximally coherent pure state $\rho=\ket{\psi}\bra{\psi}$, with $\ket{\psi}=\frac{1}{\sqrt{d}}\sum_j\ket{j}$.
\end{theorem}
\begin{proof}
This follows from the upper bound on the $f$-entropy Theorem \ref{thm:entropy-max}, and the definition of coherence
$$\tilde{\cC}_f(\rho)=\tilde{S}_f(\Delta(\rho))-\tilde{S}_f(\rho)\ . $$
For a pure state the entropy is zero, $\tilde{S}_f(\ket{\psi}\bra{\psi})=0$. The dephasing operation applied to  the state $\ket{\psi}=\frac{1}{\sqrt{d}}\sum_j\ket{j}$ gives a maximally mixed state $I/d$. The theorem follows from the fact that the entropy is maximal on maximally mixed state.
\end{proof}

\subsection{Monotonicity}
\begin{theorem}
$\cC_S$ and, in particular, $\tilde{\cC}_f$ is monotone under GIO.
\end{theorem}
\begin{proof}
Any GIO map $\Lambda$ is also SIO, and, in particular, $\Lambda$ commutes with the dephasing operation. Therefore, $\Delta(\Lambda(\rho))=\Lambda(\Delta(\rho))=\Delta(\rho)$, the last equality is due to the fact that $\Delta(\rho)\in\cI$ and $\Lambda$ as GIO preserves incoherent states. Therefore,
\begin{align}
&\cC_S(\rho)-\cC_S(\Lambda(\rho))\\
&= S(\Delta(\rho))-S(\rho) -S(\Delta(\Lambda(\rho)))+S(\Lambda(\rho))\\
&=S(\Lambda(\rho))-S(\rho)\\
&\geq 0\label{eq:mono-mono}\ ,
\end{align}
since $\Lambda$ is a CPTP map and $S$ is non-increasing under CPTP maps. For $f$-coherences, the last inequality holds to the Theorem \ref{thm:entropy-mono} since a GIO map is unital.
\end{proof}

\begin{theorem} For GIO map $\Lambda$, the equality
$$\tilde{\cC}_f(\rho)=\tilde{\cC}_f(\Lambda(\rho))$$ happens if and only if any Kraus representation of $\Lambda(\rho)=\sum_j K_j\rho K_j^*$ mush have operators $K_j=\sum_n k_{jn}\ket{n}\bra{n}$ that satisfy: for any $n,m$ such that $\bra{n}\rho\ket{m}\neq 0$, it must be  that $$\left|\sum_{j}\overline{k_{jn}}k_{jm}\right|^2=1\ .$$
\end{theorem}
\begin{proof} Similarly, to the positivity section,  equality in (\ref{eq:mono-mono}) happens if and only if there is a recovery map $\cR$ such that $\cR(\Lambda(\rho))=\rho$ and $\cR(I)=I$, \cite{HM11, HM17}. By (\ref{eq:Petz-map}), this map admits the following explicit form: denoting $\sigma=I$
$$\cR_\sigma(\omega)=\sigma^{1/2}\Lambda^*\left( \Lambda(\sigma)^{-1/2}\omega\Lambda(\sigma)^{-1/2}\right)\sigma^{1/2}\ , $$
where $\Lambda^*$ is a dual map of $\Lambda$. Since $\Lambda$ is a linear unital GIO map, we have
\begin{equation}
\cR_\sigma(\omega)= \Lambda^*(\omega)\ .
\end{equation}
Therefore, condition $\cR_\sigma(\Lambda(\rho))=\rho$ implies that
\begin{align}\label{eq:mono-equality}
\rho=\Lambda^*(\Lambda(\rho))\ .
\end{align}
Denote a Kraus representation of  $\Lambda$ as  $\Lambda(\rho)=\sum_jK_i\rho K_j^*$. From \cite{DS16}, since $\Lambda$ is GIO, any Kraus representation of $\Lambda$ has diagonal operators, i.e. each $K_j=\sum_n k_{jn}\ket{n}\bra{n}$ is diagonal in basis $\cE$. Since $\sum_j K_j^* K_j=I$, we have $\sum_j |k_{jn}|^2=1$ for every $n$. The dual map is $\Lambda^*(\rho)=\sum_j K^*_j\rho K_j$. Therefore, (\ref{eq:mono-equality}) becomes
$$\rho=\sum_{ji} K_j^*K_i \rho \left(  K_j^*K_i\right)^* \ .$$
Writing both sides in basis $\cE$ gives
\begin{align}
&\sum_{nm}\bra{n}\rho\ket{m}\ket{n}\bra{m}\\
&=\sum_{nm}\sum_{ij}\overline{k_{jn}}k_{in}k_{jm}\overline{k_{im}}\bra{n}\rho\ket{m}\ket{n}\bra{m}\\
&=\sum_{nm}\left|\sum_{j}\overline{k_{jn}}k_{jm}\right|^2\bra{n}\rho\ket{m}\ket{n}\bra{m}\ .
\end{align}
This implies that for every $n,m$ such that $\bra{n}\rho\ket{m}\neq 0$ we have
\begin{equation}\label{eq:sat}
 \left|\sum_{j}\overline{k_{jn}}k_{jm}\right|^2=1\ .
 \end{equation}
This clearly confirms that any incoherent state saturates monotonicity for GIO maps.

If $\rho$ is a coherent state, i.e. there exist $n,m$ such that $\bra{n}\rho\ket{m}\neq 0$, to saturate monotonicity the map $\Lambda$ should satisfy (\ref{eq:sat}). Note that by Cauchy-Schwarz inequality we have
$$ \left|\sum_{j}\overline{k_{jn}}k_{jm}\right|^2\leq \sum_{j}|{k_{jn}}|^2\sum_j|k_{jm}|^2=1\ . $$
The equality above happens if and only if there exists a scalar $\alpha_{nm}\in \mathbb{C}$ such that for any $j$: $k_{jn}=\alpha_{nm} k_{jm}$.
\end{proof}

Applying the strengthening of monotonicity inequality for $f$-divergences \cite{CV18}, we obtain a strengthening on the monotonicity inequality for $f$-coherence.
 
 \begin{theorem} Let $\Lambda$ be any GIO map. Let $f$ be an operator monotone 
decreasing function, and $T>0$.  Suppose for some constant $c> 0$, there is a constant $C> 0$ so that
${\rm d}t\leq  CT^{2c}{\rm d}\mu_f(t)$ for 
$t\in[T^{-1}, T]$.  Then there is an explicitly computable  constant 
$K_f(\rho)$ depending only on the smallest non-zero eigenvalue of $\rho$, $C$ and $c$, such that, 
\begin{equation}
\cC_f(\rho)-C_f(\Lambda(\rho))\geq K_f(\rho) \|\rho-\Lambda^*(\Lambda(\rho))\|_1^{4(1+c)}\ .
\end{equation}
\end{theorem}
\begin{proof}
Any GIO map $\Lambda$ commutes with the dephasing operation, therefore, 
$\Delta(\Lambda(\rho))=\Lambda(\Delta(\rho))=\Delta(\rho)$, the last equality is due to the fact that $\Delta(\rho)\in\cI$ and $\Lambda$ as GIO preserves incoherent states. Using this, we have
\begin{align}
&\cC_f(\rho)-\cC_f(\Lambda(\rho))\\
&= S_f(\Delta(\rho))-S_f(\rho) -S_f(\Delta(\Lambda(\rho)))+S_f(\Lambda(\rho))\\
&=S_f(\Lambda(\rho))-S_f(\rho)\\
&=S_f(\rho\|I/d)-S_f(\Lambda(\rho)\|I/d)\\
&=S_f(\rho\|\sigma)-S_f(\Lambda(\rho)\|\Lambda(\sigma))\ ,
\end{align}
where $\sigma=I/d$, and since $\sigma\in \cI$ and $\Lambda$ is GIO, $\Lambda(\sigma)=\sigma$. 

Applying result in \cite{CV18}, which estimates the error in the monotonicity relation for $f$-divergence, leads to the desired abound.

\end{proof}

The next theorem shows that $\tilde{\cC}_f$ is not in general monotone under SIO operations.

\begin{theorem}
If $\tilde{\cC}_f$ is monotone under all SIO, then for all states $\rho$ and $\ket{0}\in\cE$, we have
$$\tilde{\cC}_f(\rho\otimes I/d)=\tilde{\cC}_f(\rho\otimes  \ket{0}\bra{0})\ . $$
In other words, if ${\cC}_f$ is monotone under SIO, then for all states with eigenvalues $\{\lambda_j\}$ and diagonal elements $\{\chi_j\}$ in the basis $\cE$, the following holds
\begin{align}
&\sum_j \lambda_j f\left(\frac{1}{d\lambda_j}\right)-\sum_j \chi_j f\left(\frac{1}{d\chi_j}\right)\nonumber\\
&= \sum_j \lambda_j f\left(\frac{1}{d^2\lambda_j}\right)-\sum_j \chi_j f\left(\frac{1}{d^2\chi_j}\right)\ . \label{eq:sat-mono}
\end{align}
And, if  $\hat{\cC}_f$ is monotone under SIO, then for all states with eigenvalues $\{\lambda_j\}$ and diagonal elements $\{\chi_j\}$ in the basis $\cE$, the following holds
\begin{align}
&\sum_j \lambda_j f\left(\frac{1}{\lambda_j}\right)-\sum_j \chi_j f\left(\frac{1}{\chi_j}\right)\\
&= \sum_j \lambda_j f\left(\frac{d}{\lambda_j}\right)-\sum_j \chi_j f\left(\frac{d}{\chi_j}\right)\ . \label{eq:sat-mono-hat}
\end{align}
\end{theorem}
\begin{proof}
First, note that from (\ref{eq:c_f-l-c}) we have: for $\ket{0}\in\cE$, 
\begin{equation}\label{eq:C_f-0}
\cC_f(\rho\otimes \ket{0}\bra{0})=\sum_j \lambda_j f\left(\frac{1}{d^2\lambda_j}\right)-\sum_j \chi_j f\left(\frac{1}{d^2\chi_j}\right)\ , 
\end{equation}
and
\begin{equation}\label{eq:C_f-I}
\cC_f(\rho\otimes I/d)=\cC_f(\rho)=\sum_j \lambda_j f\left(\frac{1}{d\lambda_j}\right)-\sum_j \chi_j f\left(\frac{1}{d\chi_j}\right)\ . 
\end{equation}
Moreover,
\begin{equation}\label{eq:C_f-0-hat}
\hat{\cC}_f(\rho\otimes \ket{0}\bra{0})=\hat{\cC}_f(\rho)=\sum_j \lambda_j f\left(\frac{1}{\lambda_j}\right)-\sum_j \chi_j f\left(\frac{1}{\chi_j}\right)\ , 
\end{equation}
and
\begin{equation}\label{eq:C_f-I-hat}
\hat{\cC}_f(\rho\otimes I/d)=\sum_j \lambda_j f\left(\frac{d}{\lambda_j}\right)-\sum_j \chi_j f\left(\frac{d}{\chi_j}\right)\ . 
\end{equation}

Let us consider two examples of SIO$\setminus$GIO maps.\\
1. Let $\Phi(\rho)=I/d$ be the depolarizing quantum channel, which in Kraus form can be written as 
$$\Phi(\rho)=I/d=\sum_{ij=0}^{d-1}K_{ij}\rho K_{ij}^*\ ,$$
 {where} $K_{ij}=\frac{1}{\sqrt{d}}\ket{i}\bra{j}.$

Define an operation on a tensor product Hilbert space as follows
\begin{equation}
\Lambda(\omega)=\sum_{ij}(I\otimes K_{ij})\omega (I\otimes K_{ij})^*\ .
\end{equation}
Clearly, $\Lambda$ is not a GIO, since its Kraus operators are not diagonal in $\cE\otimes\cE$ basis, or since
\begin{align}
&\Lambda(\rho\otimes\ket{0}\bra{0})=\rho\otimes\Phi(\ket{0}\bra{0})=\rho\otimes I/d\\
&\neq \rho\otimes \ket{0}\bra{0}\in\cE\otimes\cE\ .
\end{align}
But $\Lambda$ is SIO, since for any $n,m$
\begin{align}
&(I\otimes K_{nm})(\Delta(\omega)(I\otimes K_{nm}^*)\\
&=\frac{1}{d}(I\otimes\ket{n}\bra{m})\left(\sum_{ij}\bra{ij}\omega\ket{ij}\ket{ij}\bra{ij}\right)(I\otimes\ket{m}\bra{n})\\
&=\frac{1}{d}\sum_{ij}\bra{ij}\omega\ket{ij}\ket{i}\bra{i}\otimes\ket{n}\bra{m}\ket{j}\bra{j}\ket{m}\bra{n}\\
&=\frac{1}{d}\sum_{i}\bra{im}\omega\ket{im}\ket{in}\bra{in}\ ,
\end{align}
and 
\begin{align}
&\Delta((I\otimes K_{nm})\omega(I\otimes K_{nm}^*))\\
&=\frac{1}{d}\sum_{ij}\bra{ij}(I\otimes\ket{n}\bra{m})\omega(I\otimes\ket{m}\bra{n})\ket{ij}\ket{ij}\bra{ij}\\
&=\frac{1}{d}\sum_{i}\bra{im}\omega\ket{im}\ket{in}\bra{in}\\
\end{align}
Therefore, $\Lambda$ is a SIO map.

For either $\cC_f$ or $\hat{\cC}_f$, consider
\begin{align}
\tilde{\cC}_f(\Lambda(\rho\otimes \ket{0}\bra{0}))&=\tilde{\cC}_f(\rho\otimes\Phi(\ket{0}\bra{0})))\\
&=\tilde{\cC}_f(\rho\otimes I/d)\label{eq:mono-1}\ .
\end{align}

2. Consider another example, let $\Psi(\rho)=\ket{0}\bra{0}$ be the erasure channel, which in Kraus form can be written as 
$$\Psi(\rho)=\ket{0}\bra{0}=\sum_{j=0}^{d-1}K_{j}\rho K_{j}^*,\qquad \text{where}\ K_{j}=\ket{0}\bra{j}\ .$$
Define an operation on a tensor product Hilbert space as follows
\begin{equation}
M(\omega)=\sum_{j}(I\otimes K_{j})\omega (I\otimes K_{j})^*\ .
\end{equation}
Clearly, $M$ is not a GIO, since its Kraus operators are not diagonal in $\cE\otimes\cE$ basis, or since
\begin{align}
&M(\rho\otimes I/d)=\rho\otimes \Psi(I/d)=\rho\otimes  \ket{0}\bra{0}\\
&\neq \rho\otimes I/d\in\cE\otimes\cE\ .
\end{align}
But $M$ is SIO, since  for any $n$,
\begin{align}
&(I\otimes K_{n})\Delta(\omega)(I\otimes K_{n}^*)\\
&=\sum_{ij}\bra{ij}\omega\ket{ij} (I\otimes\ket{0}\bra{n})\ket{ij}\bra{ij}(I\otimes\ket{n}\bra{0})\\
&=\sum_{i}\bra{in}\omega\ket{in} \ket{i}\bra{i}\otimes \ket{0}\bra{0}\ .
\end{align}
and 
\begin{align}
&\Delta\left((I\otimes K_{n})\omega(I\otimes K_{n}^*)\right)\\
&=\sum_{ij}\bra{ij}(I\otimes \ket{0}\bra{n})\omega(I\otimes \ket{n}\bra{0})\ket{ij}\ket{ij}\bra{ij}\\
&=\sum_{i}\bra{in}\omega\ket{in}\ket{i0}\bra{i0}\ .
\end{align}
Therefore, $M$ is an SIO map.

For either $\cC_f$ or $\hat{\cC_f}$, consider
\begin{align}
\tilde{\cC}_f(M(\rho\otimes I/d))&=\tilde{\cC}_f(\rho\otimes\Psi(I/d)))\\
&=\tilde{\cC}_f(\rho\otimes \ket{0}\bra{0})\label{eq:mono-2}\ .
\end{align}

Now, compare (\ref{eq:mono-1}) and (\ref{eq:mono-2}). In order for monotonicity of $f$-coherence  to hold under all SIO, there must be an equality
$$\tilde{\cC}_f(\rho\otimes I/d)=\tilde{\cC}_f(\rho\otimes  \ket{0}\bra{0})\ . $$
Invoking (\ref{eq:C_f-0}-\ref{eq:C_f-I-hat}) we have the result stated in the theorem.
\end{proof}

Note that both (\ref{eq:sat-mono}) and (\ref{eq:sat-mono-hat}) hold for the logarithmic function $f(x)=-\log(x)$, but fail for the power function  $f(x)=\frac{1}{1-\alpha}(1-x^{1-\alpha})$. This is in line with the fact that the relative entropy of coherence is monotone under SIO, and it shows that Tsallis coherence fails monotonicity for SIO.

\subsection{Strong monotonicity}

\begin{theorem}
$f$-coherences $\tilde{\cC}_f$ saturate strong monotonicity for convex mixtures of diagonal unitaries. Therefore, $\tilde{\cC}_f$ saturates strong monotonicity under GIO in two- and three-dimensions.
\end{theorem}
\begin{proof}
Consider an example of GIO, which is a probabilistic mixture of diagonal unitaries: for some $\alpha_j>0$, s.t. $\sum_j \alpha_j=1$, define
$$\Lambda(\rho)=\sum_j \alpha_j U_j\rho U^*_j\ , $$
where for some $\rho_{jn}$ the unitaries $U_j$ are diagonal in $\cE$, i.e.
$$U_j=\sum_n e^{i\phi_{jn}}\ket{n}\bra{n}\ . $$

In \cite{DS16} it has been shown that all GIO are of such form for dimensions two and three, but it is no longer the case for higher dimensions. 

Note that for $\sigma=I$ or $\sigma=I/d$ and for all unitaries $U$, we have
\begin{equation}\label{eq:s-f-unitary}
S_f(U\rho U^*\|\sigma)=S_f(\rho\|\sigma)\ . 
\end{equation}
Taking $U_j$ diagonal in $\cE$ above, it follows that
$$\Delta(U_j\rho U_j^*)=\Delta(\rho)\ . $$
Therefore, $\tilde{\cC}_f$ saturates the strong monotonicity under convex mixtures of diagonal unitaries:
\begin{align}
&\sum_j\alpha_j\tilde{\cC}_f(U_j \rho U_j^*)\\
&=\sum_j\alpha_j \left[S_f(U_j\rho U_j^*\|\sigma)-S_f(\Delta(U\rho U^*)\|\sigma)\right]\\
&=\sum_j\alpha_j\left[S_f(\rho\|\sigma)-S_f(\Delta(\rho)\|\sigma)\right]\\
&=\tilde{\cC}_f(\rho )\ .
\end{align}
\end{proof}

\begin{remark}
Expanding the set of operations to include all unitaries (not necessarily diagonal in $\cE$), forces $\tilde{\cC}_f$ to be invariant under all unitaries if it is monotone under them. This results from the following observation: if $\tilde{\cC}_f$ is monotone under all unitaries $U$ and all states $\rho$, then, since (\ref{eq:s-f-unitary}) holds, it must be that
$$S_f(\Delta(U\rho U^*)\|\sigma)\geq S_f(\Delta(\rho)\|\sigma)\ . $$
But taking a unitary $V=U^*$ and an initial state $\omega=U\rho U^*$ above, results in the opposite inequality:
\begin{align}
&S_f(\Delta(V\omega V^*)\|\sigma)=S_f(\Delta(\rho)\|\sigma)\\
&\geq S_f(\Delta(U\rho U^*)\|\sigma)=S_f(\Delta(\omega)\|\sigma)\ . 
\end{align}
Therefore, the above inequality must be equality, which makes $\tilde{\cC}_f$ invariant under unitaries.
\end{remark}

\begin{theorem}
For any pure state $\rho$, the $f$-coherences are strongly monotone under GIO maps in any finite dimension.
\end{theorem}
\begin{proof}
Let us denote $\sigma=I$ or $\sigma=I/d$ depending on the $f$-coherence we are considering. For a GIO map $\Lambda$ with Kraus operators $K_j$, denote 
$$p_j=\Tr K_j\rho K_j^*, \ \rho_j=\frac{1}{p_j} K_j\rho K_j^*\ . $$
For a pure state $\rho$, states $\rho_j$ are also pure. Therefore,
\begin{align}
&\tilde{\cC}_f(\rho)-\sum_jp_j\tilde{\cC}_f(\rho_j)\\
&=\sum_j p_j {S}_f(\Delta(\rho_j)\|\sigma)-{S}_f(\Delta(\rho)\|\sigma)\ .
\end{align}
Since any GIO map is an SIO map as well, it follows that
$$\Delta(\rho_j)=\frac{1}{p_j}K_j\Delta(\rho) K_j^*\ . $$

Dephased state $\Delta(\rho)$ is diagonal in $\cE$ basis with eigenvalues $\chi_j$, i.e. $\Delta(\rho)=\sum_j \chi_j\ket{j}\bra{j}$. The $f$-divergence is
$${S}_f(\Delta(\rho)\| I)=\sum_n\chi_nf\left( \frac{1}{\chi_n}\right)\ . $$

Kraus operators of GIO map are diagonal is $\cE$ basis, $K_j=\sum_n k_{jn}\ket{n}\bra{n}$, with $\sum_j|k_{jn}|^2=1$ for all $j$. Then 
$$K_j\Delta(\rho) K_j^*=\sum_n\chi_n|k_{jn}|^2\ket{n}\bra{n}\ . $$
And
$$\sum_j p_j {S}_f(\Delta(\rho_j)\|I)=\sum_{jn}\chi_n |k_{jn}|^2 f\left(\frac{p_j}{\chi_n|k_{jn}|^2} \right)\ .
$$
Since $f$ is convex, we have for every $n$:
\begin{align}
\sum_{j}|k_{jn}|^2 f\left(\frac{p_j}{\chi_n|k_{jn}|^2} \right)&\geq f\left(\sum_j\frac{p_j}{\chi_n}  \right)\\
&=f\left( \frac{1}{\chi_n}\right)\ .
\end{align}
Similarly, 
$$
{S}_f(\Delta(\rho)\| I/d)=\sum_n\chi_nf\left( \frac{1}{d\chi_n}\right)\ ,
$$
and 
$$\sum_j p_j {S}_f(\Delta(\rho_j)\|I/d)=\sum_{jn}\chi_n |k_{jn}|^2 f\left(\frac{p_j}{d\chi_n|k_{jn}|^2} \right)\ .
$$
Because $f$ is convex, for any $n$:
\begin{align}
\sum_{j}|k_{jn}|^2 f\left(\frac{p_j}{d\chi_n|k_{jn}|^2} \right)&\geq f\left(\sum_j\frac{p_j}{d\chi_n}  \right)\\
&=f\left( \frac{1}{d\chi_n}\right)\ .
\end{align}
And thus, 
$\sum_j p_j {S}_f(\Delta(\rho_j)\|\sigma)\geq {S}_f(\Delta(\rho)\|\sigma). $
Which implies that for any pure state $\rho$, the $f$-coherence is strongly monotone under GIO:
$$\tilde{\cC}_f(\rho)\geq\sum_jp_j\tilde{\cC}_f(\rho_j)\ . $$

\end{proof}

Data sharing not applicable to this article as no datasets were generated or analysed during the current study.

\vspace{0.3in}
\textbf{Acknowledgments.}  A. V. is supported by NSF grants DMS-1812734 and DMS-2105583.

\end{multicols}
\end{document}